\documentclass[ copyright]{eptcs}

\usepackage{fixltx2e} 

\usepackage[utf8]{inputenc}
\usepackage[T1]{fontenc}   

\usepackage{breakurl}        

\usepackage[cmex10]{amsmath}

\usepackage{array, tikz}
\usepackage{subfig,tabularx}
\usepackage{xspace}
\usepackage{cite}
\usepackage{epsfig}
\usepackage{multicol}
\usepackage{marvosym}
\usepackage[english]{babel}
\usepackage{amssymb}

\usepackage{amsfonts,stmaryrd,}
\usepackage{enumerate}
\usepackage{amsthm}

\usepackage{graphicx}

\usepgflibrary{shapes}

\newtheorem{definition}{Definition}

\newtheorem{proposition}{Proposition}

\renewcommand{\lambda}{\ensuremath{\pi}} 
\newcommand{\cg}{\ensuremath{\text{NATS}}\xspace} 
\newcommand{\cgp}{\ensuremath{\mathit{NATSs}}\xspace} 

\newcommand{\slc}{\ensuremath{\text{USL}}\xspace}
\newcommand{\usl}{\ensuremath{\text{USL}}\xspace}
\newcommand{\slnc}{\ensuremath{\text{SL}}\xspace}

\newcommand{\cando}[1]{\langle\negthinspace\langle #1
  \rangle\negthinspace\rangle}

\newcommand{\canndo}[1]{\llbracket #1 \rrbracket}

\renewcommand{\mu}{\ensuremath{\alpha}\xspace}

\newcommand{\stra}{\ensuremath{\mathit{Strat}}\xspace}

\newcommand{\de}{\ensuremath{\mathit{Ch}}\xspace}
\newcommand{\st}{s.t.\@\xspace}

\newcommand{\ut}{\ensuremath{\ \mathbf{U}\ }}

\renewcommand{\circ}{\ensuremath{\mathbf{X}\ }}
\newcommand{\eq}{\ensuremath{\mathbf{E}}\xspace}
\newcommand{\aq}{\ensuremath{\mathbf{A}}\xspace}

\newcommand{\track}{\ensuremath{\mathit{track}_{\mathcal{M}}}\xspace}
\newcommand{\last}{\ensuremath{\mathit{last}}}
\newcommand{\plan}{\ensuremath{\chi}\xspace}
\newcommand{\dom}{\ensuremath{\mathit{dom}}\xspace}
\newcommand{\out}{\ensuremath{\mathit{out}}\xspace}
\newcommand{\strat}{\ensuremath{\mathit{Strat}}\xspace}
\newcommand{\ct}{\ensuremath{\kappa}\xspace}

\newcommand{\valu}{\ensuremath{\mathit{v}}\xspace}
\newcommand{\atoms}{\ensuremath{\mathit{At}}\xspace}
\newcommand{\agent}{\ensuremath{\mathit{Ag}}\xspace}
\newcommand{\choic}{\ensuremath{\mathit{Ch}}\xspace}

\newcommand{\pspace}{\textbf{P{\small SPACE}}}

\newcommand{\bind}[2]{\ensuremath{ (#1\vartriangleright #2)}\xspace}
\newcommand{\debind}[2]{\ensuremath{ (#1\ntriangleright #2)}\xspace}
\newcommand{\debibi}[2]{\ensuremath{[#1\vartriangleright #2]}\xspace}

\newcommand{\p}{\ensuremath{\lambda}\xspace}
\newcommand{\tr}{\ensuremath{\tau}\xspace}

\author{Christophe Chareton \quad\quad Julien Brunel \quad\quad David
  Chemouil \institute{Onera -- The French Aerospace Lab \\F-31055
    Toulouse, France} \email{firstname.lastname@onera.fr} }

\providecommand{\event}{SR 2013}

\def\titlerunning{Towards an Updatable Strategy Logic} 

\def\authorrunning{Ch. Chareton, J. Brunel \& D. Chemouil}

\begin{document}
\title{\titlerunning}
\maketitle

\begin{abstract}

This article is about temporal multi-agent logics.  Several of these
formalisms have been already presented (ATL-ATL*, ATL$_{\mathit{sc}}$,
SL). They enable to express the capabilities of agents in a system to
ensure the satisfaction of temporal properties. Particularly, SL and
ATL$_{\mathit{sc}}$ enable several agents to interact in a context
mixing the different strategies they play in a semantical game. We
generalize this possibility by proposing a new formalism, Updating
Strategy Logic (\slc). In \slc, an agent can also refine its own
strategy. The gain in expressive power rises the notion of
\emph{sustainable capabilities} for agents.

\slc is built from SL. It mainly brings to SL the two following
modifications: semantically, the successor of a given state is not
uniquely determined by the data of one choice from each
agent. Syntactically, we introduce in the language an operator, called
an \emph{unbinder}, which explicitly deletes the binding of a strategy
to an agent. We show that USL is strictly more expressive than SL. 
\end{abstract}


\section{Introduction}\label{intro}

Multi-agent logics are receiving growing interest in contemporary
research. Since the seminal work of Rajeev Alur, Thomas A. Henzinger,
and Orna Kupferman \cite{ahk}, one major and recent direction (ATL
with Strategy Context \cite{brihaye2009atl,atlscconcur, dacostalopes}, Strategy
Logic (presented first in \cite{krish2010} and then extended in \cite{vardi, varditr}) aims at contextualizing the statements
of capabilities of agents.

Basically, multi-agent logics enable assertions about the capability of
agents to ensure temporal properties. Thus, ATL-ATL$^{*}$ \cite{ahk} appears as a
generalization of CTL-CTL$^{*}$, in which the path quantifiers $\eq$
and $\aq$ are replaced by \emph{strategy quantifiers}. Strategy
quantifiers (the existential $\cando{A}$ and the universal
$\canndo{A}$) have a (coalition of) agent(s) as parameter.
$\cando{A}\varphi$ means that agents in $A$ can act so as to ensure
the satisfaction of temporal formula $\varphi$. It is interpreted in
\emph{Concurrent Game Structures} (CGS), where agents can make choices
influencing the execution in the system. Formula $\cando{A}\varphi$ is
true if agents in $A$ have a strategy so that if playing it they force
the execution to satisfy $\varphi$, whatever the other agents do.

A natural question is: how to interpret the imbrication of several
quantifiers? Precisely, in the interpretation of such formula as
\[
\psi_{1}:=\cando{a_{1}}\Box(\varphi_{1}\wedge\cando{a_{2}}\Box
\varphi_{2})
\]
 (where $\Box \varphi$ is the temporal operator meaning
\emph{$\varphi$} is always true, and $a_{1}$ and $a_{2}$ are agents),
is the evaluation of $\varphi_{2}$ made in a context that takes into
account both the strategy quantified in $\cando{a_{1}}$ and the
strategy quantified in $\cando{a_{2}}$?

In ATL-ATL$^{*}$, 
only $a_{2}$ is bound: subformula
$\cando{a_{2}}\Box\varphi_{2}$ is true iff $a_{2}$ may ensure
$\Box \varphi_{2}$, whatever the other agents do.  Then
$\cando{a_{2}}$ stands for three successive operations: First, each agent is unbound from its current strategy, then  an existential quantification is made for strategy $\sigma$.
At last, $a_{2}$ is bound to strategy $\sigma$.

ATL$_{\text{sc}}$~\cite{brihaye2009atl,atlscconcur, dacostalopes}, while keeping the ATL
syntax, adapts the semantics in order to interpret formulas in a context which
stores strategies introduced by earlier quantifiers. 

\emph{Strategy Logic} (\slnc~\cite{vardi, varditr}) is
another interesting proposition, which distinguishes between the
quantifications over strategies and their bindings to agents. The
operator $\cando{a}$ is split into two different operators: a
quantifier over strategies ($\cando{x}$, where $x$ is a strategy
variable) and a binder ($(a,x)$, where $a$ is an agent) that stores
into a context the information that $a$ plays along the strategy
instantiating variable $x$ (let us write it $\sigma_{x}$ in the
remaining of this paper).The ATL formula $\psi_{1}$ syntactically
matches the \slnc:
\[
\psi_{2}:=
\cando{x_{1}}(a_{1},x_{1})\Box(\varphi_{1}\wedge\cando{x_{2}}(a_{2},x_{2})\Box
\varphi_{2})
\]
 In $\psi_{2}$, when evaluating $\Box \varphi_{2}$,
$a_{1}$ remains bound to strategy $\sigma_{x_{1}}$ except if $a_{1}$ and
$a_{2}$ are the same agent. If they are the same, the binder $(a_{2},x_{2})$
unbinds $a$ from its current strategies before binding her to
$\sigma_{x_{2}}$.

In this paper we present \slc, a logic obtained from \slnc by making
explicit the unbinding of strategies and allowing new bindings without
previous unbinding. For that, we introduce an explicit unbinder
$\debind{a}{x}$ in the syntax (and the binder in \slc is written
\bind{a}{x}) and we interpret \slc in models where the choices of
agents are represented by the set of potential successors they enable from the current
state.  When there is no occurrence of an unbinder, each agent remains
bound to her current strategies. Then different strategies can combine
together even for a single agent, provided that they are
\emph{coherent}, which means they define choices in non-empty intersection
(the notion is formally defined in Sect.~\ref{sem}).

The main interest in such introduction is to distinguish between cases
where an agent composes strategies together and situations where she
revokes a current strategy for playing an other one.
If $a_{1}$ and $a_{2}$ are the same agents, then $\psi_{2}$ is written in \slnc:
\[
\psi_{3}:=
\cando{x_{1}}(a,x_{1})\Box(\varphi_{1}\wedge\cando{x_{2}}(a,x_{2})\Box
\varphi_{2}),
\]

 which syntactically matches the \slc:
\[
\psi_{4}:=\cando{x_{1}}\bind{a}{x_{1}}\Box(\varphi_{1}\wedge\cando{x_{2}}\bind{a}{x_{2}}\Box\varphi_{2})
\]

In $\psi_{3}$, subformula $\cando{x_{2}}(a,x_{2})\Box \varphi_{2}$
states that $a$ can adopt a new strategy that ensures $\Box
\varphi_{2}$, no matter if it is coherent with the strategy
$\sigma_{x_{1}}$ previously adopted. In $\psi_{4}$, both strategies
must combine coherently together. In natural language $\psi_{4}$
states that $a$ can ensure $\varphi_{1}$ and leave open the
possibility to ensure $\varphi_{2}$ in addition.  The equivalent of
$\psi_{3}$ in \slc is actually not $\psi_{4}$ but
\[
\psi_{5}:=\cando{x_{1}}\bind{a}{x_{1}}\Box(\varphi_{1}\wedge\cando{x_{2}}\debind{a}{x_{1}}\bind{a}{x_{2}}\Box\varphi_{2})
\]
 There indeed,
in subformula $\debind{a}{x_{1}}\bind{a}{x_{2}}\Box\varphi_{2}$, $a$
is first unbound from $\sigma_{x_{1}}$ and then bound to
$\sigma_{x_{2}}$.

A consequence of considering these compositions of
strategies is the expressiveness of \emph{sustainable capabilities} of
agents. Let us now consider the \slc formula:
\[
\psi_{6}:= \cando{x_{1}}\bind{a}{x_{1}}\Box(\cando{x_{2}}\debind{a}{x_{1}}\bind{a}{x_{2}}\circ p)
\]
 There the binder $\bind{a}{x_{2}}$ is used with the unbinder $\debind{a}{x_{1}}$, so that $\psi_{6}$ is equivalent to the \slnc: 
\[
\psi_{7}:= \cando{x_{1}}(a,x_{1})\Box(\cando{x_{2}}(a, x_{2})\circ p)
\]
 It states that $a$ can remain capable to perform the condition expressed
 by $\circ p$ when she wants. But in case she actually performs it,
 the formula satisfaction does not require that she is still capable to
 perform it. The statement holds in state $s_{0}$ in structure
 $\mathcal{M}_{1}$ with single agent $a$. See Fig.\ref{red1}, where choices
 are defined by the set of transitions they enable. Since
 $\mathcal{M}_{1}$ interprets SL formulas with only agent $a$, the
 choices for $a$ are deterministic: let $s,s'$ be two states and $c$ a
 choice, then the transition from $s$ to $s'$ is labelled with $c$ iff
 $\{s'\}$ is a choice for $a$ at $s$.  Indeed, by always playing
 choice $c_{1}$, $a$ remains in state $s_{0}$, 
 where she can change her mind to ensure $p$. But if she chooses
 to reach $p$, she can do it only  by moving to state $s_{1}$ and then to state
 $s_{2}$. Doing so, she loses her capability to ensure $\circ p$ at any
 time.  The only way for her to maintain her capability to reach $p$ is
 to always avoid it, her capability is not sustainable.
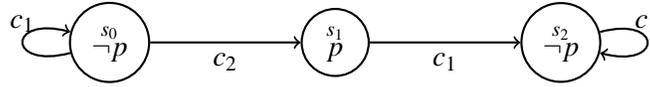
\begin{figure}[!h]\centering
\begin{tikzpicture}
\tikzstyle{st}=[circle, 
thick, draw];
\node[st](1) at(0,0){$\stackrel{s_{0}}{\neg p}$};
\node[st](2) at(3,0){$\stackrel{s_{1}}{p}$};
\node[st](3) at(6,0){$\stackrel{s_{2}}{\neg p}$};
\draw[thick, ->](1)to node[below]{$c_{2}$}(2);
\draw[thick, ->](2)to node[below]{$c_{1}$}(3);
\draw[loop left,thick, ->](1)to node[above]{$c_{1}$}(1);
\draw[loop right,thick, ->](3)to node[above]{$c_{1}$}(1);
\end{tikzpicture}
\caption{Structure $\mathcal{M}_{1}$\label{red1}}
\end{figure}

A more game theoretical view is to consider strategies as
commitments. In such view, by adopting a strategy, $a$ adopts a
behavior that holds in the following execution, as far as it is not
explicitly deleted. Formula
\[
\psi_{8}:= \label{f6}\cando{x_{1}}\bind{a}{x_{1}}\Box(\cando{x_{2}}\bind{a}{x_{2}}\circ
p)
\]
 is the counterpart of formula $\psi_{7}$ with such interpretation
of composing strategies for a single agent.  If $a$ plays
$\sigma_{x_{2}}$, it must be coherently with $\sigma_{x_{1}}$. Thus,
$\psi_{8}$ is false in structure $\mathcal{M}_{1}$, since $a$ cannot
achieve $p$ more than once. 

Formula $\psi_{8}$ distinguishes between structures $\mathcal{M}_{1}$
and $\mathcal{M}_{2}$ from Fig.\ref{red2} ( Note that in this second
structure the choices are not deterministic: from a given state a
choice may be compatible with several potential successors). In
$\mathcal{M}_{2}$, $\psi_{8}$ is true at $s_{0}$ since the strategy
\emph{always play $c_{1}$} ensure the execution to remain in state
$s_{0}$ or $s_{1}$ and is always coherent with strategy \emph{play
  $c_{2}$ first and then always play $c_{1}$}, which ensures $\circ p$
from states $s_{0}$ and $s_{1}$. What is at stake with it is the
difference between \emph{sustainable capabilities} and \emph{one shot
  capabilities}. Formulas $\psi_{7}$ and $\psi_{8}$ both formalize the
natural language proposition \emph{$a$ can always achieve $p$}. One
shot capability ($\psi_{7}$) means she can achieve it once for all and
choose when. Sustainable capability ($\psi_{8}$) means she can achieve
it and choose when without affecting nor losing this capability for
the future.

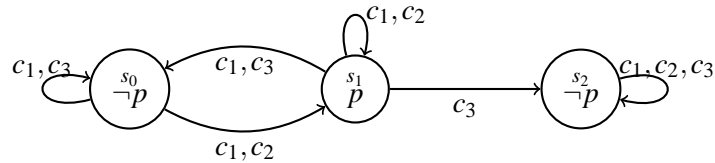
\begin{figure}[!h]\centering
\begin{tikzpicture}
\tikzstyle{st}=[circle, 
thick, draw];
\node[st](1) at(0,0){$\stackrel{s_{0}}{\neg p}$};
\node[st](2) at(3,0){$\stackrel{s_{1}}{p}$};
\node[st](3) at(6,0){$\stackrel{s_{2}}{\neg p}$};
\draw[thick, ->, bend right](1)to node[below]{$c_{1}, c_{2}$}(2);
\draw[thick, ->, bend right](2)to node[below]{$c_{1}, c_{3}$}(1);

\draw[thick, ->](2)to node[below]{$c_{3}$}(3);
\draw[loop left,thick, ->](1)to node[above]{$c_{1}, c_{3}$}(1);
\draw[loop above,thick, ->](2)to node[right]{$c_{1}, c_{2}$}(2);
\draw[loop right,thick, ->](3)to node[above]{$c_{1}, c_{2}, c_{3}$}(1);
\end{tikzpicture}
\caption{Structure $\mathcal{M}_{2}$ \label{red2}}
\end{figure}

In Sect.\ref{epower}, we compare the expressive power of \slnc and \slc by use of formula $\psi_{9}$, obtained from $\psi_{7}$ by adding to $a$ the sustainable capability to ensure $\circ \neg p$:
\[\psi_{9}:= \cando{x}\bind{a}{x}\Box(\cando{x_{0}}\bind{a}{x_{0}}\circ p
\wedge\cando{x_{0}}\bind{a}{x_{0}}\circ \neg p)\] $\psi_{9}$ states
that $a$ has sustainable capability to decide whether $p$ or $\neg p$
holds at next state. We say that $a$ has \emph{sustainable} control on
property $p$: she is sustainably capable to decide the truth value of $p$.

The main purposes of \slc are  to give a formalism for
the composition of strategies and  to unify it with
the classical branching-time mechanisms of strategy revocation. So, 
 both treatments can be combined in a single formalism. 
In the remaining of this paper we define \slc  syntax and semantics, and we introduce the comparison of its expressive power with that of \slnc.


\section{Syntax and semantics}\label{sem}
In this section we present the syntax and semantics of \slc, together
with the related definitions they require. The \slc formulas distinguish
between \emph{path} and \emph{state} formulas.

\begin{definition}  Let 
$\agent$ be a set of agents,
$\atoms$ a set of propositions and $X$ a set of
variables, \slc($\agent,\atoms,X$) is given by the following grammar:
\begin{itemize}
\item State formulas: 
$ \varphi::= p \mid \neg\varphi\mid
\varphi\wedge\varphi\mid\cando{x}\varphi\mid \bind{A}{x}\psi\mid
\debind{A}{x}\psi$
 \item Path formulas:  
$\psi ::=\varphi\mid \neg \psi\mid\psi\wedge\psi\mid\psi\ut\psi\mid\circ\psi$
\end{itemize}
where $p\in\atoms, A\subseteq \agent, x\in X$. 
\end{definition}
These formulas hold a notion of \emph{free} variable that is similar
to that in \cite{vardi, varditr}: an atom has an empty set of free variables, a
binder adds a free variable to the set of free variables of its direct
subformula and a quantifier deletes it.  Upon formulas on this
grammar, those that can be evaluated with no context are the
\emph{sentences}. They are formulas with empty set of free variables,
which means each of their bound variables is previously quantified.
We now come to the definitions for \slc  semantics.
\begin{definition} \label{cgsc}
A Non-deterministic Alternating Transition System (\cg) is a tuple \linebreak
$\mathcal{M} = \langle \agent, M, \atoms, \valu, \de \rangle$ where:
\begin{itemize}
\item $M$ is a set of states, called the domain of the \cg, $\atoms$ is
  the set of atomic propositions and $\valu$ is a
  valuation function, from $M$ to $\mathcal{P}(\atoms)$.
\item \de : $ \agent\times M \to \mathcal{P}(\mathcal{P}(M))$ is a
  choice function mapping a pair $( \mathit{agent,
    state})$ to a non-empty family of choices of possible next
  states. It is such that for every state $s\in M$ and for every agents $a_{1}$ and $a_{2}$ in $\agent$, for every $c_{1}\in\choic(a_{1},s)$ and $c_{2}\in\choic(a_{2},s), c_{1}\cap c_{2}\neq \emptyset$.
\end{itemize}
\end{definition}

We call a finite sequence of states in $M$ a \emph{track} $\tr$. The
last element of a track $\tau$ is denoted by \last$(\tau)$. 
The set
of tracks that are possible in $\mathcal{M}$ is denoted by $\track:
\tau = s_{0}s_{1}\dots s_k \in \track$ iff for every $i<k$, for every $a\in \agent$, there is 
$c_{a}\in \mathcal{P}(M)$ \st $c_{a}\in \choic(a,s_{i})$ and $s_{i+1}\in
c_{a}$.  Similarly, an infinite sequence of states such that all its prefixes are in $\track$ is called a
\emph{path} (in $\mathcal{M}$).

\begin{definition}[Strategies and coherence]
 A \emph{strategy} is a function $\sigma$ from $\agent\times \track$
 to $\mathcal{P}(M)$ such that for all $(a, \tr)\in \agent\times
 \track, \sigma(a,\tr)\in \de(a, \last(\tr))$.  By
 extension, we write $\sigma(A, \tau)$ for $\bigcap_{a\in A} \sigma(a,
 \tau)$ for every $A\subseteq\agent$. Two strategies $\sigma_{1}$ and
 $\sigma_{2}$ are \emph{coherent} iff for all $(a, \tau)$ in
 $\agent\times \track, \sigma_{1}(a, \tau) \cap \sigma_{2}(a,
 \tau)\neq\emptyset$. In this case, we also say that $\sigma_{1}(a, \tau)$ and
 $\sigma_{2}(a, \tau)$ are \emph{coherent choices}.
\end{definition}

 A \emph{commitment} \ct is a finite sequence upon
 $(\mathcal{P}(\agent)\times X)$, representing the active bindings.
 An \emph{assignment} $\mu$ is a partial function from $X$ to
 \strat. A \emph{context} \plan is a pair of an assignment and a
 commitment.  Note that an agent can appear several times in a
 commitment. Furthermore commitments store the order in which pairs
 $(A,x)$ are introduced. Therefore our notion of contexts differs from
 the notion of \emph{assignments} that is used in \slnc~\cite{vardi,
   varditr}.

 A context defines a function from $\track$ to $\mathcal{P}(M)$. We
 use the same notation for the context itself and its induced
 function. Let $\ct_{\emptyset}$ be the empty sequence upon
 $(\mathcal{P}(\agent)\times X)$, then:
\begin{itemize}
\item $(\mu,\ct_{\emptyset})(\tr) = M$
\item $(\mu,(A,x))(\tr)=$ 
\begin{itemize}
\item $\bigcap_{a\in A}\mu(x)(a,\tr)$ if $A \neq
 \emptyset$
\item else $M$
\end{itemize}
\item $(\mu,\ct\cdot(A,x))(\tr) =$
\begin{itemize}
\item $(\mu,\ct)(\tr)\cap(\mu,(A,x))(\tr)$ if this intersection is not
 empty.
\item  otherwise (which means the context induces contradictory choices),
 $ (\mu,\ct)(\tr)$ . 
\end{itemize}
\end{itemize}
Now we can define the
 outcomes of a context \plan, \out$(\plan)$: let $\p = \p_{0}, \p_{1},
 \dots$ be an infinite sequence over $M$, then $\p\in \out(s, \plan)$
 iff $\p$ is a path in $\mathcal{M}$, $s = \p_{0}$ and for every $n
 \in \mathbb{N}$, $\p_{n+1}\in\plan(\p_{0} \dots \p_{n})$.

\begin{definition}[Strategy and assignment translation] \label{transla}Let
  $\sigma$ be a strategy and $\tr$ be a track. Then $\sigma^{\tr}$ is
  the strategy \st for every $\tr'\in \track$ , $\sigma^{\tr}(\tr') =
  \sigma(\tr\tr')$. The notion is extended to an assignment: for every
  $\mu, \mu^{\tr}$ is the assignment with domain equal to that of $\mu$
  and \st for every $x\in \dom(\mu), \mu^{\tr}(x)=(\mu(x))^{\tr}$
\end{definition}

We also define the following transformations of commitments and
assignments. Given a commitment $\ct$, coalitions $A$ and $B$, a strategy
variable $x$, an assignment $\mu$ and a strategy $\sigma$:
\begin{itemize}
\item $\ct[A\rightarrow x] = \ct\cdot\bind{A}{x}$
\item $((B,x)\cdot\ct)[A\nrightarrow x] = (B\backslash A,
x)\cdot(\ct[A\nrightarrow x])$ and $\ct_{\emptyset}[A\nrightarrow x]= \ct_{\emptyset}$
\item $ \mu[x\to\sigma]$ is
the assignment with domain $\dom(\mu)\cup\{x\}$ \st $\forall y \in
\dom(\mu)\backslash\{x\}, \mu[x\to\sigma](y)= \mu(y)$ and
$\mu[x\to\sigma](x)=\sigma$
\end{itemize}

\begin{definition}[Satisfaction relation] 
Let $\mathcal{M}$ be a \cg, then for every assignment \mu, commitment \ct,
state $s$ and path $\lambda$:
\begin{itemize}
\item State formulas:
\begin{itemize}
\item $\mathcal{M}, \mu,\ct, s\models p$ iff $p\in \valu(s)$, with
  $p \in \atoms$
\item $\mathcal{M}, \mu, \ct, s\models \neg \varphi$ iff it is
  not true that $\mathcal{M}, \mu,\ct, s\models\varphi$
\item $\mathcal{M}, \mu,\ct, s\models \varphi_{1}\wedge
  \varphi_{2}$ iff $\mathcal{M}, \mu,\ct, s\models \varphi_{1}$
  and $\mathcal{M}, \mu,\ct, s\models \varphi_{2}$
\item $\mathcal{M}, \mu,\ct,s\models \cando{x}\varphi$ iff
  there is a strategy $\sigma\in \stra$ \st $\mathcal{M}, \mu[x\to
    \sigma],\ct, s\models\varphi$
\item $\mathcal{M},\mu,\ct , s\models \bind{A}{x}\varphi$ iff  for every
  $\lambda$ in $\out(\mu, \ct[A \rightarrow x]), \mathcal{M}, \mu,
  \ct[A \rightarrow x], \lambda\models\varphi$
\item $\mathcal{M},\mu,\ct , s\models \debind{A}{x}\varphi$ iff for
  all $\lambda$ in $\out(\mu, \ct[A \nrightarrow x]), \mathcal{M},
  \mu, \ct[A \nrightarrow x], \lambda\models\varphi$
\end{itemize}
\item Path formulas :
\begin{itemize}
\item $\mathcal{M},\mu,\ct , \lambda\models\varphi$ iff
  $\mathcal{M},\mu,\ct , \lambda_{0}\models\varphi$, for every state
  formula $\varphi$
\item $\mathcal{M}, \mu, \ct, \lambda\models \neg \psi$ iff it is
  not true that $\mathcal{M}, \mu,\ct, \lambda\models\psi$
\item $\mathcal{M}, \mu,\ct, \lambda\models \psi_{1}\wedge
  \psi_{2}$ iff $\mathcal{M}, \mu,\ct, \lambda\models \psi_{1}$
  and $\mathcal{M}, \mu,\ct, \lambda\models \psi_{2}$
\item $\mathcal{M}, \mu,\ct,\lambda\models \circ\psi$ iff 
  $\mathcal{M}, \mu^{\lambda_{0}},\ct,\lambda^1\models
  \psi$.
\item $\mathcal{M}, \mu,\ct,\lambda\models
  \psi_{1}\ut\psi_{2}$ iff there is $i\in \mathbb{N}$ \st
  $\mathcal{M}, \mu^{\lambda_{0}\dots\lambda_{i-1}},\ct,\lambda^{i}\models
  \psi_{2}$ and for every $0\leq j< i, \mathcal{M},
  \mu^{\lambda_{0}\dots\lambda_{j-1}},\ct,\lambda^{j}\models \psi_{1}$
\end{itemize}
\end{itemize}
Let $\mu_{\emptyset}$ be the unique assignment with empty domain. Let
$\varphi$ be a sentence in \slc$(\agent,\atoms,X)$. Then $\mathcal{M}, s
\models \varphi$ iff $\mathcal{M}, \mu_{\emptyset},\ct_{\emptyset}
\models \varphi$.
\end{definition}

Let us give the following comment over these definitions: for every
context \plan = $(\mu, \ct)$, the definition of \out(\plan) ensures that
the different binders encoded in \plan compose their choices together,
\emph{as far as possible}. In case two contradictory choices from an
agent are encoded in the context, the priority is given to the first
binding that was introduced in this context (the left most binding in the
formula). This guarantees that a formula requiring the composition of
two contradictory strategies is false. For example, suppose that
$\cando{x_{1}}\bind{a}{x_{1}}\varphi_{1}$ and
$\cando{x_{2}}\bind{a}{x_{2}}\varphi_{2}$ are both true in a state of
a model, and suppose that strategies  $\sigma_{x_1}$ and $\sigma_{x_2}$
necessarily rely on contradictory choices of $a$ (this means that $a$
cannot play in a way that ensures both $\varphi_{1}$ and
$\varphi_{2}$). Then,
$\cando{x_{1}}\bind{a}{x_{1}}(\varphi_{1}\wedge\cando{x_{2}}\bind{a}{x_{2}}\varphi_{2})$
is false in the same state of the same model. If the priority was
given to the most recent binding (right most binding in the formula),
the strategy $\sigma_{x_1}$ would be revoked and the formula would be
satisfied.

\section{Comparison with \slnc~\cite{vardi, varditr}}\label{epower}
\slnc syntax can be basically described from SL by deleting the use of the
unbinder. Furthermore, the binders are limited to sole agents and are
written $(a,x)$ instead of $\bind{a}{x}$.  \slc appears to be more
expressive than \slnc \cite{vardi, varditr}. More precisely, \slnc can
be embedded in \slc, while $\psi_{9}$ is not expressible in \slnc,
even by extending its semantics to non-deterministic models.  Here we
give the three related propositions. By lack of space, the proofs are
only sketched in this article. Detailed proofs of these propositions
can be found in \cite{tr2}.  Note that, since \slnc is strictly more
expressive than ATL$_{\text{sc}}$ \cite{dacostalopes}, the following
results also hold for comparing \slc with ATL$_{\text{sc}}$.
\begin{proposition} 
\label{sec:comparison-with-slnc}
There is an embedding of \slnc into \slc. 
\end{proposition}
\begin{proof}[Proof (Sketch)]
The embedding consists in a parallel transformation from \slnc models
and formulas to that of \slc. The transformation preserves the
satisfaction relation.  The differences between \slnc and \slc lie
both in the definition of strategies in \slnc semantics and the
difference of interpretation for the binding operator. The first is
treated by defining an internal transformation for \slnc. By this
transformation, the constraints of agents playing the same choices,
issued from SL actions framework, are expressed in the syntax. Then we
define a new operator in \slc that is equivalent to \slnc binding, and
show the equivalence: the operator \debibi{a}{x} is an abbreviation
for a binder \bind{a}{x} preceded by the set of unbinders
$\debind{a}{x_{i}}$, one for every variable $x_{i}$ in the language.
\end{proof}
\begin{proposition}
\label{sec:comparison-with-slnc-1}A model is said \emph{deterministic} if the successor of a state is uniquely determined by  one choice for every agent. Then, sustainable control is not expressible  over deterministic models, neither in \slnc nor in \slc.
\end{proposition}
\begin{proof}[Proof (Sketch)]
One checks that for every deterministic \cg $\mathcal{M}$, for any state $s$ of $\mathcal{M}$, $\mathcal{M}, s\nvDash \psi_{9}$. Proposition~\ref{sec:comparison-with-slnc} then straightly brings proposition~\ref{sec:comparison-with-slnc-1}
\end{proof}
\begin{proposition}
Sustainable control is not expressible  in \slnc interpreted over \cgp.
\end{proposition}

\begin{proof}[Proof (Sketch)]
The proof uses a generalization of SL  semantics over \cgp. Its definition is in \cite{tr2} and holds, for example, the following cases:
\begin{itemize}

\item $\mathcal{M}, \mu,\ct,s \models_{\cg} \circ\varphi$ iff for every $\lambda\in \out(s,(\mu, \ct)),
  \mathcal{M}, \mu^{\lambda_{0}},\ct,\lambda_{1}\models_{\cg}
  \varphi$
\item $\mathcal{M}, \mu,\ct,\lambda\models_{\cg} \varphi_{1}\ut\varphi_{2}$ iff
  for every $\lambda\in \out(s,(\mu, \ct))$, there is $i\in \mathbb{N}$
  \st $\mathcal{M}, \mu^{\lambda_{0}\dots\lambda_{i-1}},\ct,\lambda^{i}\linebreak\models_{\cg} \varphi_{2}$ and for
  all $0\leq j\leq i, \mathcal{M}, \mu^{\lambda_{0}\dots\lambda_{i-1}},\ct,\lambda^{j}\models_{\cg}
  \varphi_{1}$.
\item $\mathcal{M}, \mu,\ct,s\models_{\cg} \cando{x}\varphi$ iff
  there is a strategy $\sigma\in \stra$ \st $\mathcal{M}, \mu[x\to
    \sigma],\ct, s\models_{\cg}\varphi$.
\item $\mathcal{M},\mu,\ct , s\models_{\cg} (a,x)\varphi$ iff  $ \mathcal{M}, \mu,  \ct[x\backslash \ct(a)], s\models_{\cg}\varphi$.
\end{itemize}
where $\ct[x\backslash \ct(a)]$ designates the context obtained from $\ct$ by replacing every $(a,y)$ in it by $(a, x)$.

Formula  $\psi_{9}$ states that $a$ can always control whether $p$ or not. 
Suppose there is a formula $\varphi$ in \slnc equivalent to $\psi_{9}$
and let us call \emph{existential} a formula in \slnc in which every occurrence
of $\cando{x}$ is under an even number of quantifiers. If $\varphi$
is existential then under binary trees it is equivalent to a formula in $\Sigma_{1}^{1}$ (the fragment of second order logic with only existential set quantifiers). 

We now consider a set of formulas $\{\Gamma_{i}\}_{i\in \mathbb{N}}$,
each one stating that $a$ can choose $i$ times between $p$ and $\neg
p$. The set $\{\Gamma_{i}\}_{i\in \mathbb{N}}$ is defined by induction
over $i$:
\begin{itemize}
\item $\Gamma_{0}:=
\cando{x}(a,x)\Box(\cando{x_{0}}(a,x_{0})\circ p
\wedge\cando{x_{0}}(a,x_{0})\circ \neg p)$
\item for all $i\in\mathbb{N}, \Gamma{i+1} = \Gamma_{i}[
  p \wedge\Box(\cando{x_{i+1}}(a,x_{i+1})\circ p
  \wedge\cando{x_{i+1}}(a, x_{i+1})\circ \neg p \backslash p]\newline[\neg p
  \wedge\Box(\cando{x_{i+1}}(a,x_{i+1})\circ p \wedge\cando{x_{i+1}}(a,x_{i+1})\circ
  \neg p)\backslash \neg p]$.
\end{itemize}
where the notation $\theta_{1}[\theta_{2}\backslash\theta_{3}]$ designates
the formula obtained from $\theta_{1}$ by replacing any occurrence of
subformula $\theta_{3}$ in it by $\theta_{2}$.  $\{\Gamma_{i}\}_{i\in
  \mathbb{N}}$ is equivalent to $\varphi$. A compactness argument shows
that it is not equivalent to a formula in $\Sigma_{1}^{1}$ under
binary trees, hence $\varphi$ is not an existential formula.  Then, we
notice that $\varphi$ is true in structures where, from any state, $a$ can ensure any
labelling of sequences over $p$. So, if $\varphi$ has a
subformula $(a,x)\psi$ where $x$ is universally quantified, $\psi$
must be equivalent to $\Box(p\vee\neg p)$. Then, by iteration, $\varphi$ is
equivalent to an existential formula in \slnc. Hence a contradiction.
\end{proof}


\section{Conclusion}

In this article we defined a strategy logic with updatable
strategies. By updating a strategy, agents remain playing along it but
add further precision to their choices. This mechanism enables to
express such properties as sustainable capability and sustainable
control. To the best of our knowledge, this is the first proposition
for expressing such properties. Especially, the comparison introduced
with SL in this article could be adapted to ATL with Strategy Context
\cite{brihaye2009atl}.

The revocation of strategies is also questioned in
\cite{agotnes2007}. The authors propose a formalism with definitive
strategies, that completely determine the behaviour of agents.  They
also underline the difference between these strategies and revocable
strategies in the classical sense. We believe that updatable
strategies offer a synthesis between both views: updatable strategies
can be modified without being revoked.

Strategies in \slc can also be explicitly revoked. This idea  is already
present in \cite{brihaye2009atl} with the operator $\cdot\rangle A
\langle\cdot$. But the  operator $\langle\cdot
A \cdot\rangle$ also implicitly unbinds current  strategy for agents in $A$ before binding them a new strategy. Thus it prevents agents from  updating their
strategy or composing several strategies.

Further study perspectives about \slc mainly concern the model
checking. Further work will provide it with a proof of non elementary
decidability, adapted from the proof in \cite{vardi}. We are also
working on a semantics for \slc under memory-less strategies
and \pspace\ algorithm for its model-checking. Satisfiability problem
should also be addressed. Since \slnc SAT problem is not decidable,
similar result is expectable for \slc.  Nevertheless,
decidable fragments of \usl may be studied in the future,
in particular by  following the directions given in \cite{slconcur}.


\bibliographystyle{eptcs}
\bibliography{main}
\end{document}